\documentclass[submission,copyright,creativecommons]{eptcs}
\providecommand{\event}{QPL 2022} 

\usepackage[utf8]{inputenc}
\usepackage[english]{babel}
\usepackage{iftex}
\usepackage{amsmath,amsthm,amssymb}
\usepackage{qcircuit}
\usepackage{physics}
\usepackage{caption}
\usepackage{subcaption}
\usepackage[ruled]{algorithm2e}
\usepackage{appendix}
\usepackage{floatrow}
\usepackage{graphicx}
\usepackage{tikz}
\usepackage{svg}

\newcommand{\B}{\mathbb{B}}
\newcommand{\N}{\mathbb{N}}

\newcommand{\BB}[1][n]{\B^{\B^{#1}}}

\newtheorem{lemma}{Lemma}[section]

\ifpdf
  \usepackage{underscore}         
  \usepackage[T1]{fontenc}        
\else
  \usepackage{breakurl}           
\fi

\title{Tunable Quantum Neural Networks in the QPAC-Learning Framework}
\author{Viet Pham Ngoc
\institute{Imperial College London\\London, United Kingdom}
\email{viet.pham-ngoc17@imperial.ac.uk}
\and
David Tuckey 
\institute{Imperial College London\\London, United Kingdom}
\email{david.tuckey17@imperial.ac.uk}
\and 
Herbert Wiklicky
\institute{Imperial College London\\London, United Kingdom}
\email{h.wiklicky@imperial.ac.uk}
}
\def\titlerunning{TNN in the QPAC-Learning Framework}
\def\authorrunning{V. Pham Ngoc, D. Tuckey \& H. Wiklicky}
\begin{document}
\maketitle

\begin{abstract}
In this paper, we investigate the performances of tunable quantum neural networks in the Quantum Probably Approximately Correct (QPAC) learning framework. Tunable neural networks are quantum circuits made of multi-controlled $\mathbf{X}$ gates. By tuning the set of controls these circuits are able to approximate any Boolean functions. This architecture is particularly suited to be used in the QPAC-learning framework as it can handle the superposition produced by the oracle. In order to tune the network so that it can approximate a target concept, we have devised and implemented an algorithm based on amplitude amplification. The numerical results show that this approach can efficiently learn concepts from a simple class.
\end{abstract}

\section{Introduction}
Machine learning is believed to be an application of quantum computing that will yield promising results. It is the reason why, in the past years, significant efforts have been put into developing machine learning techniques suited to quantum computers. One interesting approach to this task consists in adapting existing classical techniques \cite{Rebentrost2018, Tacchino2019, Schuld2014, Lloyd2013} to take advantage of quantum properties and gain some speed-up. 

Introduced by \cite{Valiant1984}, probably approximately correct (PAC) learning provides a mathematical framework to analyse classical machine learning techniques. This framework revolves around the existence of an oracle that provides samples drawn from the instance space. These examples are then used as examples for an algorithm to learn a target concept. The efficiency of this learning algorithm can then be characterised thanks to its sampling complexity. Given the momentum gained by quantum techniques for machine learning, it seems natural for a quantum equivalence of PAC-learning to be introduced. This is exactly what was done in \cite{Bshouty1998} with QPAC-learning. In this framework, instead of providing samples from the instance space, the oracle generates a superposition of all the examples. Similarly to PAC-learning, QPAC-learning can be used to evaluate the learning efficiency of quantum algorithm.

In this paper we are using this framework to study a learning algorithm based on quantum amplitude amplification \cite{Brassard2002}. This fundamental quantum procedure allows us both to compare the error rate to a threshold as well as increase the probability of measuring the errors produced by the learner. These measured errors are then used to tune the learner in order to decrease the error rate. In our case, this learner is a tunable quantum neural network \cite{Pham2020} which is essentially a quantum circuit made of multi-controlled $\mathbf{X}$ gates. We prove that this setup is able to learn efficiently by implementing it and testing it against a simple class of concepts.


\section{Related Works}\label{section:related}
Studied in  \cite{Pham2020}, the tunable quantum neural networks (TNN) are an interesting kind of quantum circuits in the sense that they naturally express Boolean functions hence are able to approximate any target function. At the core of this concept is the fact that a Boolean function can be expanded into a polynomial form called the algebraic normal form (ANF). This polynomial form can then easily be transposed into a quantum circuit. 

Let $u = u_0 \ldots u_{n-1} \in \B^n$, we denote $\mathbf{1}_u = \{i \in [0,n-1] \ | \ u_i = 1\}$ and for $x \in \B^n$, $x^u = \prod_{i \in \mathbf{1}_u}x_i$. Let $f \in \BB$ then its algebraic normal form is:
\begin{equation}\label{eq:anf}
    f(x) = \bigoplus_{u \in \B^n}{\alpha_ux^u}
\end{equation}
Where $\alpha_u \in \{0,1\}$ denotes the absence or the presence of the monomial $x^u$ in the expansion. As an example, let us consider the function $g \in \BB[2]$ defined by $g(x_0,x_1) = 1 \oplus x_1 \oplus x_0.x_1$. With the previously introduced notations, its ANF can be written: $g(x) = x^{00} \oplus x^{01} \oplus x^{11}$ with $\alpha_{00}=\alpha_{01}=\alpha_{11}=1$ and $\alpha_{10}=0$. Using this form, any Boolean function can be expressed with a quantum circuit made of multi-controlled $\mathbf{X}$ gates. 

Let $f \in \BB$ with ANF as in (\ref{eq:anf}) and consider a quantum circuit of $n+1$ qubits with the first $n$ qubits being denoted $q_0, \ldots, q_{n-1}$ and the last one being the ancillary qubit, denoted $a$. Now for $u \in \B^n$ such that $\alpha_u = 1$ in the ANF of $f$, place an $\mathbf{X}$ gate controlled by the qubits $\{q_i \ | \ i \in \mathbf{1}_u\}$ on the ancillary qubit. Let $\mathbf{QC}$ be the resulting circuit, then it is such that for $x \in \B^n$ and $b \in \B$:
$$
    \mathbf{QC}\ket{x}\ket{b} = \ket{x}\ket{b \oplus f(x)}
$$
Let us consider the function $g \in \BB[2]$ introduced previously, then the circuit pictured in Figure \ref{fig:ex_anf} is expressing $g$.
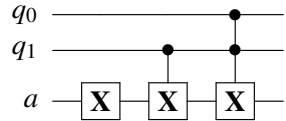
\begin{figure}[h]
    \centering
    $$
    \Qcircuit @C=1em @R=1em{
        \lstick{q_0} & \qw & \qw & \ctrl{1} & \qw\\
        \lstick{q_1} & \qw & \ctrl{1} & \ctrl{1} & \qw\\
        \lstick{a} & \gate{\mathbf{X}} & \gate{\mathbf{X}} & \gate{\mathbf{X}} & \qw
    }
    $$
    \caption{Circuit expressing $g(x) = 1 \oplus x_1 \oplus x_0.x_1$.}
    \label{fig:ex_anf}
\end{figure}
\\
A tunable neural network is then a quantum circuit of this type for which the set of controlled $\mathbf{X}$ gates can be tuned so that the expressed function is an approximation of a target function. Now let $f \in \BB$ and consider the network $\mathbf{TNN}(f)$ expressing $f$, then for a superposition of the form $\ket{\phi} = \sum_{x \in \B^n}{d_x\ket{x}}\ket{0}$ we have:
$$
    \mathbf{TNN}(f)\ket{\phi} = \sum_{x \in \B^n}{d_x\ket{x}}\ket{f(x)}
$$
Which is reminiscent of the output from the oracle encountered in the QPAC-learning framework.  

QPAC-Learning has been introduced in \cite{Bshouty1998} and is a quantum version of the work presented in \cite{Valiant1984}. In this framework, we call $\mathcal{C} \subseteq \BB$ a class of concept and the aim is to learn $c \in \mathcal{C}$. A probability distribution $D$ is placed over $\B^n$ and for a hypothesis $h \in \BB$ the error is defined by 
\begin{equation}\label{eq:err}
    err_D(h,c) = P_{x \sim D}(h(x) \neq c(x))
\end{equation}
For $0 < \epsilon < \frac{1}{2}$ and $0 < \delta < \frac{1}{2}$, the goal is then to produce $h \in \mathcal{C}$ such that:
\begin{equation}\label{eq:pac}
    P\left(err_D(h,c) < \epsilon \right) > 1 - \delta
\end{equation}
A class of concept $\mathcal{C}$ is said to be PAC-learnable with an algorithm $A$ if for $c \in \mathcal{C}$, for all distribution $D$ and for $0 < \epsilon < \frac{1}{2}$, $0 < \delta < \frac{1}{2}$, $A$ will output a hypothesis $h \in \BB$ verifying (\ref{eq:pac}).

During the learning process, we are given access to an oracle $\mathbf{EX}(c,D)$ such that each call to this oracle will produce the superposition $\ket{\Psi(c,D)}$:
$$
    \ket{\Psi(c,D)} = \sum_{x \in \B^n}{\sqrt{D(x)}\ket{x}\ket{c(x)}}
$$
For the rest of the paper we will assume that $\mathbf{EX}(c,D)$ is a quantum gate such that:
$$
    \mathbf{EX}(c,D)\ket{0} = \ket{\Psi(c,D)}
$$
With these notations, for a concept $c$ and a hypothesis $h$ we have:
$$
     err_D(h,c) = \sum_{h(x) \neq c(x)}{D(x)}
$$

A learning algorithm is then said to be an efficient PAC-learner when the number of calls to $\mathbf{EX}$ that are necessary to attain (\ref{eq:pac}) is polynomial in $\frac{1}{\epsilon}$ and $\frac{1}{\delta}$.

The learning algorithm presented in this paper is based on quantum amplitude amplification. Introduced by \cite{Brassard2002} the quantum amplitude amplification algorithm is a generalisation of Grover's algorithm \cite{Grover1996}. Let $G$ be the set of states we want to measure and $\mathcal{A}$ be a quantum circuit such that $\mathcal{A}\ket{0} = \ket{\Phi}$ with:
$$
    \ket{\Phi} = \sum_{x \notin G}{d_x\ket{x}} + \sum_{x \in G}{d_x\ket{x}}
$$
This state can be rewritten as
\begin{equation}\label{eq:Phi}
    \ket{\Phi} = \cos(\theta)\ket{\Phi_B} + \sin(\theta)\ket{\Phi_G}
\end{equation}
With $\theta \in \left[0,\frac{\pi}{2}\right]$, $\cos(\theta) = \sqrt{\sum_{x \notin G}{|d_x|^2}}$, $\ket{\Phi_B} = \frac{1}{\cos(\theta)}\sum_{x \notin G}{d_x\ket{x}}$, $\sin(\theta) = \sqrt{\sum_{x \in G}{|d_x|^2}}$ and $\ket{\Phi_G} = \frac{1}{\sin(\theta)}\sum_{x \in G}{d_x\ket{x}}$. We also have $\braket{\Phi_B}{\Phi_G} = 0$

Now let $\mathcal{X}_G$ be such that:
$$
    \mathcal{X}_G\ket{x} = \begin{cases}
                    -\ket{x} & \text{if } x \in G \\
                    \phantom{-}\ket{x} & \text{otherwise}
                    \end{cases}
$$
Then we have $\mathcal{X}_G\ket{\Phi_G} = -\ket{\Phi_G}$ and $\mathcal{X}_G\ket{\Phi_B} = \ket{\Phi_B}$. We also define $\mathcal{X}_0 = \mathbf{I}-2\ket{0}\bra{0}$ then by denoting $\mathbf{Q} = -\mathcal{A}\mathcal{X}_0\mathcal{A}^{-1}\mathcal{X}_G$, $\mathbf{Q}$ is the diffusion operator and acts on $\ket{\Phi}$ in the following way:
$$
    \mathbf{Q}^m\ket{\Phi} = \cos\left((2m+1)\theta\right)\ket{\Phi_b} + \sin\left((2m+1)\theta\right)\ket{\Phi_g} \text{ for }m \in \N
$$
Suppose we know $\theta$, then by choosing $m$ such that $(2m+1)\theta \approx \frac{\pi}{2}$, for example $m = \left\lfloor \frac{1}{2}\left(\frac{\pi}{2\theta}-1\right) \right\rfloor$, we are ensured that when measuring, a state in $G$ will almost certainly be measured.

This amplitude amplification procedure is at the heart of numerous quantum amplitude estimation algorithms \cite{Brassard2002,Nakaji2020, Aaronson2020}. In this setup, a state of a form similar to that of (\ref{eq:Phi}) is given but $\theta \in \left[0,\frac{\pi}{2}\right]$ is unknown and these algorithms seek to approximate $\sin(\theta)$. Where \cite{Brassard2002} makes use of quantum Fourier transforms an integrant part of the algorithm, \cite{Nakaji2020,Aaronson2020} do without while still maintaining quantum speed-up. If $a=\sin(\theta) \in [0,1]$ is the amplitude to be evaluated and $\epsilon,\delta>0$, then these algorithms output an estimation $\tilde{a}$ such that $P(|\tilde{a}-a|< a\epsilon)>1-\delta$.


\section{TNN in the QPAC-Learning Framework}
As said previously, tunable networks are particularly well suited to be employed in the QPAC-learning framework. Let $\mathcal{C} \subseteq \BB$ be a class of concept and $c \in \mathcal{C}$. Let $D$ be a probability distribution over $\B^n$ and $\mathbf{EX}(c,D)$ the oracle such that:
$$
    \mathbf{EX}(c,D)\ket{0} = \ket{\Psi(c,D)} = \sum_{x \in \B^n}{\sqrt{D(x)}\ket{x}\ket{c(x)}}
$$
Now let $\mathbf{TNN}(h)$ be a tunable neural network expressing $h \in \BB$ in its current state. Then
we have:
\begin{align*}
    \ket{\Phi}  &= \mathbf{TNN}(h)\ket{\Psi(c,D)} \\
                &= \sum_{x \in \B^n}{\sqrt{D(x)}\ket{x}\ket{c(x) \oplus h(x)}}\\
                &= \sum_{h(x)=c(x)}{\sqrt{D(x)}\ket{x}\ket{0}} + \sum_{h(x) \neq c(x)}{\sqrt{D(x)}\ket{x}\ket{1}}
\end{align*}
The corresponding circuit is given in Figure \ref{fig:circuit_Phi}.
\begin{figure}[h]
    \centering
    $$
    \Qcircuit @C=1em @R=1em{
    \lstick{\ket{0}^{\otimes n}} & \multigate{1}{\mathbf{EX}(c,D)} & \multigate{1}{\mathbf{TNN}(h)} & \qw \\
    \lstick{\ket{0}^{\phantom{\otimes n}}} & \ghost{\mathbf{EX}(c,D)} & \ghost{\mathbf{TNN}(h)} & \qw
    }
    $$
    \caption{Circuit resulting in the state $\ket{\Phi}$.}
    \label{fig:circuit_Phi}
\end{figure}
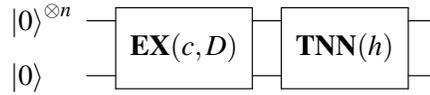
\\
The state $\ket{\Phi}$ can then be rewritten as:
$$
    \ket{\Phi} = \cos(\theta_e)\ket{\phi_g}\ket{0} + \sin(\theta_e)\ket{\phi_e}\ket{1}
$$
With $\braket{\phi_g,0}{\phi_e,1} = 0$. This state is similar to the one encountered in the amplitude amplification procedure. Because $\sin(\theta_e) = \sqrt{\sum_{h(x) \neq c(x)}{D(x)}}$, we also have
$$
    \sin^2(\theta_e) = \sum_{h(x) \neq c(x)}{D(x)} = err_D(h,c)
$$
In order to use tunable neural networks in the QPAC-learning framework we thus need a mean to estimate $\sin^2(\theta_e)$ and compare it to $0< \epsilon < 1/2$. If the error is smaller than $\epsilon$, stop. Otherwise, tune the network so that it expresses another hypothesis $h'$ and repeat. The algorithm proposed in this paper does just that by using amplitude amplification.


\section{Tuning Algorithm}
At its core, the algorithm is similar to quantum amplitude estimation, the difference being that it compares the unknown amplitude $err_D(h,c)$ to $\epsilon$ instead of estimating the amplitude to a relative error of $\epsilon$. Let $\theta_{err}$ and $\theta_{\epsilon}$ such that $\sin^2(\theta_{err}) = err_D(h,c)$ and $\sin^2(\theta_{\epsilon}) = \epsilon$. Because $D$ is a distribution we can take $\theta_{err} \in \left[0,\frac{\pi}{2}\right]$ and because $0<\epsilon<\frac{1}{2}$ we can assume $\theta_{\epsilon} \in \left]0,\frac{\pi}{4}\right[$. There are then two possible cases: either $err_D(h,c) \geq \epsilon$ or $err_D(h,c) < \epsilon$. Because $\sin$ is increasing on $\left[0,\frac{\pi}{2}\right]$, these translate into $\theta_{err} \geq \theta_{\epsilon}$ or $\theta_{err} < \theta_{\epsilon}$ respectively. Now let $m_{\text{max}}$ be the smallest integer such that $(2m_{\text{max}}+1)\theta_{\epsilon} \geq \frac{\pi}{4}$, we introduce this following simple lemma:
\begin{lemma}\label{lem:inf}
    Let $\theta \in \left[0, \frac{\pi}{2}\right]$. If for all $m \leq m_{\text{max}}$, we have $\sin^2\left((2m+1)\theta\right) < \frac{1}{2}$ then $\theta < \theta_{\epsilon}$
\end{lemma}
\begin{proof}
    We show this by contraposition.\\
    Let $\theta \geq \theta_{\epsilon}$.\\
    Suppose $\theta \in \left[\frac{\pi}{4},\frac{\pi}{2}\right]$, then taking $m=0 \leq m_{\text{max}}$, we have $\sin^2(\theta) \geq \frac{1}{2}$.\\
    Now suppose that $\theta \in \left[0,\frac{\pi}{4}\right[$. By definition we have $m_{\text{max}} = \left\lceil\frac{1}{2}\left(\frac{\pi}{4\theta_{\epsilon}}-1\right)\right\rceil$ and we denote $m = \left\lceil\frac{1}{2}\left(\frac{\pi}{4\theta}-1\right)\right\rceil$, then $m \leq m_{\text{max}}$. Let $0 \leq \alpha < 1$ such that $m = \frac{1}{2}\left(\frac{\pi}{4\theta}-1\right) + \alpha$. Then we have $(2m+1)\theta = \frac{\pi}{4} + 2 \alpha \theta$ and $\frac{\pi}{4} \leq (2m+1)\theta \leq \frac{3\pi}{4}$. This leads to $\sin^2((2m+1)\theta) \geq \frac{1}{2}$.\\ 
    Hence Lemma \ref{lem:inf}.
\end{proof}
This result is illustrated in Figures \ref{fig:sup} and \ref{fig:inf}. 
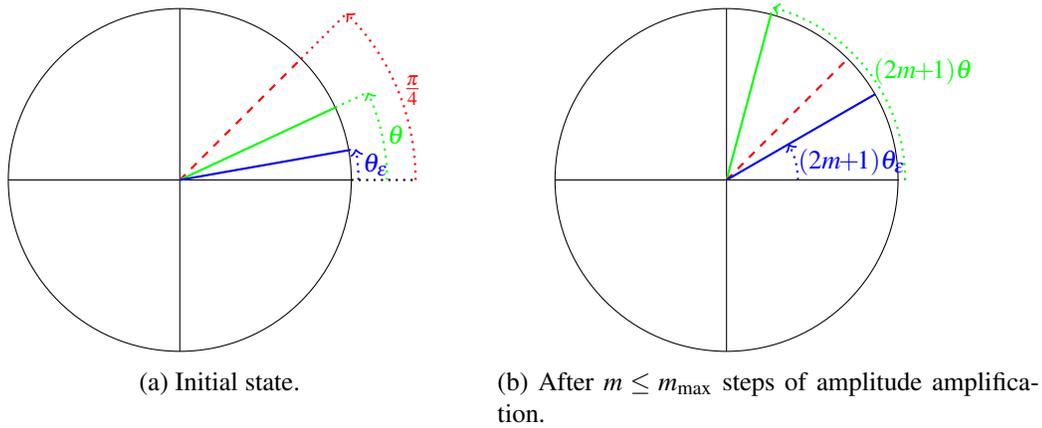
\begin{figure}[h]
    \centering
    \begin{subfigure}[t]{0.45\textwidth}
        \centering
        \begin{tikzpicture}[scale=1.9]
            \draw (0,0) circle [radius=1.2];
            \draw (-1.2,0) -- (1.2,0);
            \draw[thick, dotted] (1.2,0) -- (1.65,0);
            \draw (0,-1.2) -- (0,1.2);
            \draw[red, thick, dashed] (0,0) -- ({1.2*cos(45)},{1.2*sin(45)});
            \draw[red, thick, dotted] ({1.2*cos(45)},{1.2*sin(45)}) -- ({1.65*cos(45)},{1.65*sin(45)});
            \draw[red, thick, dotted, ->] (1.65,0) arc (0:45:1.65) node [label={[font=\small, red, right, midway,xshift=-2]$\frac{\pi}{4}$}]{};
            \draw[green, thick] (0,0) -- ({1.2*cos(25)},{1.2*sin(25)});
            \draw[green, thick, dotted] ({1.2*cos(25)},{1.2*sin(25)}) -- ({1.45*cos(25)},{1.45*sin(25)});
            \draw[green, thick, dotted, ->] (1.45,0) arc (0:25:1.45) node [label={[font=\small, green, right, midway,xshift=-2]$\theta$}]{};
            \draw[blue, thick] (0,0) -- ({1.2*cos(10)},{1.2*sin(10)});
            \draw[blue, thick, dotted] ({1.2*cos(10)},{1.2*sin(10)}) -- ({1.25*cos(10)},{1.25*sin(10)});
            \draw[blue, thick, dotted, ->] (1.25,0) arc (0:10:1.25) node [label={[font=\small, blue, right, midway,xshift=-2]$\theta_{\epsilon}$}]{};
        \end{tikzpicture} 
        \caption{Initial state.}
    \end{subfigure}
    \begin{subfigure}[t]{0.45\textwidth}
        \centering
        \begin{tikzpicture}[scale=1.9]
            \draw (0,0) circle [radius=1.2];
            \draw (-1.2,0) -- (1.2,0);
            \draw[dotted] (1.2,0) -- (1.25,0);
            \draw (0,-1.2) -- (0,1.2);
            \draw[green, thick] (0,0) -- ({1.2*cos(75)},{1.2*sin(75)});
            \draw[green, thick, dotted] ({1.2*cos(75)},{1.2*sin(75)}) -- ({1.25*cos(75)},{1.25*sin(75)});
            \draw[green, thick, dotted, ->] (1.25,0) arc (0:75:1.25) node [label={[font=\small, green, right, midway,xshift=-2]$(2m\!\!+\!\!1)\theta$}]{};
            \draw[red, thick, dashed] (0,0) -- ({1.2*cos(45)},{1.2*sin(45)});
            \draw[blue, thick] (0,0) -- ({1.2*cos(30)},{1.2*sin(30)});
            \draw[blue, thick, dotted, ->] (0.5,0) arc (0:30:0.5) node [label={[font=\small, blue, right, midway,xshift=-3]$(2m\!\!+\!\!1)\theta_{\epsilon}$}]{};
        \end{tikzpicture}
        \caption{After $m \leq m_{\text{max}}$ steps of amplitude amplification.}
    \end{subfigure}
    \caption{Case where $\theta \geq \theta_{\epsilon}$.}
    \label{fig:sup}
\end{figure}
\begin{figure}[h]
    \centering
    \begin{subfigure}[t]{0.45\textwidth}
        \centering
        \begin{tikzpicture}[scale=1.9]
            \draw (0,0) circle [radius=1.2];
            \draw (-1.2,0) -- (1.2,0);
            \draw[thick, dotted] (1.2,0) -- (1.65,0);
            \draw (0,-1.2) -- (0,1.2);
            \draw[red, thick, dashed] (0,0) -- ({1.2*cos(45)},{1.2*sin(45)});
            \draw[red, thick, dotted] ({1.2*cos(45)},{1.2*sin(45)}) -- ({1.65*cos(45)},{1.65*sin(45)});
            \draw[red, thick, dotted, ->] (1.65,0) arc (0:45:1.65) node [label={[font=\small, red, right, midway,xshift=-2]$\frac{\pi}{4}$}]{};
            \draw[blue, thick] (0,0) -- ({1.2*cos(10)},{1.2*sin(10)});
            \draw[blue, thick, dotted] ({1.2*cos(10)},{1.2*sin(10)}) -- ({1.45*cos(10)},{1.45*sin(10)});
            \draw[blue, thick, dotted, ->] (1.45,0) arc (0:10:1.45) node [label={[font=\small, blue, right, midway,xshift=-2]$\theta_{\epsilon}$}]{};
            \draw[green, thick] (0,0) -- ({1.2*cos(7)},{1.2*sin(7)});
            \draw[green, thick, dotted] ({1.2*cos(7)},{1.2*sin(7)}) -- ({1.25*cos(7)},{1.25*sin(7)});
            \draw[green, thick, dotted, ->] (1.25,0) arc (0:7:1.25) node [label={[font=\small, green, right, midway,xshift=-2, yshift=1]$\theta$}]{};
        \end{tikzpicture} 
        \caption{Initial state.}
    \end{subfigure}
    \begin{subfigure}[t]{0.45\textwidth}
        \centering
        \begin{tikzpicture}[scale=1.9]
            \draw (0,0) circle [radius=1.2];
            \draw (-1.2,0) -- (1.2,0);
            \draw[dotted] (1.2,0) -- (1.25,0);
            \draw (0,-1.2) -- (0,1.2);
            \draw[blue, thick] (0,0) -- ({1.2*cos(50)},{1.2*sin(50)});
            \draw[blue, thick, dotted] ({1.2*cos(50)},{1.2*sin(50)}) -- ({1.25*cos(50)},{1.25*sin(50)});
            \draw[blue, thick, dotted, ->] (1.25,0) arc (0:50:1.25) node [label={[font=\small, blue, right, midway,xshift=-3]$(2m_{\text{max}}\!\!+\!\!1)\theta_{\epsilon}$}]{};
            \draw[red, thick, dashed] (0,0) -- ({1.2*cos(45)},{1.2*sin(45)});
            \draw[green, thick] (0,0) -- ({1.2*cos(35)},{1.2*sin(35)});
            \draw[green, thick, dotted, ->] (0.3,0) arc (0:35:0.3) node [label={[font=\small, green, right, midway,xshift=-2, yshift=1]$(2m_{\text{max}}\!\!+\!\!1)\theta$}]{};
        \end{tikzpicture}
        \caption{After $m_{\text{max}}$ steps of amplitude amplification.}
    \end{subfigure}
    \caption{Case where $\theta < \theta_{\epsilon}$.}
    \label{fig:inf}
\end{figure}
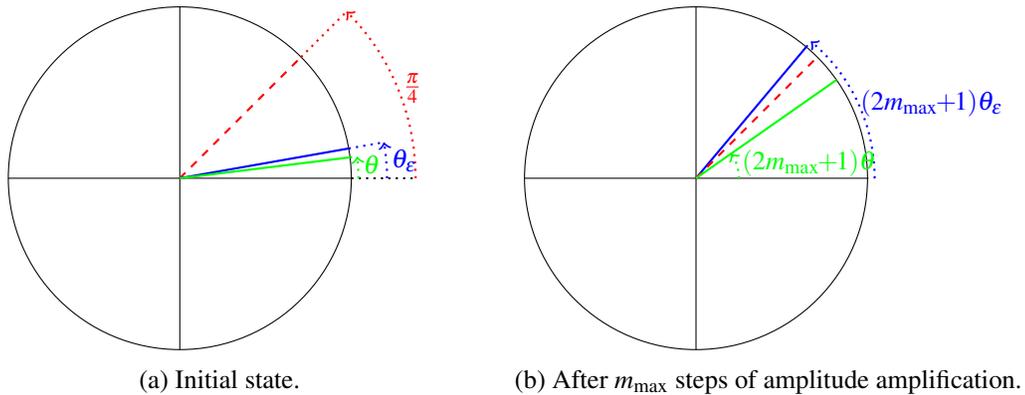
\\
The tuning algorithm works as follow. After each update of the tunable network, the error is compared to $\epsilon$ thanks to Lemma \ref{lem:inf}: after $m_0 < m_{\text{max}}$ steps of amplification, the resulting state is measured $N$ times and the number $S$ of measurements for which the ancillary qubit is in state $\ket{1}$ is counted. \\
If $S > \frac{N}{2}$, then the error is greater than $\epsilon$ and the network is updated using the measurements on the first $n$ qubits. If on the other hand $S \leq \frac{N}{2}$, the process is repeated with $m_0+1 \leq m_{\text{max}}$ steps of amplification.\\
The algorithm stops when the network reaches a state such that even after $m_{\text{max}}$ steps of amplification we have $S \leq \frac{N}{2}$. It is thus necessary to choose $N$ such that when the algorithm stops the error is most probably lower than $\epsilon$.

To avoid the  angle $(2m_{\text{max}}+1)\theta_{\epsilon}$ from overshooting $\frac{\pi}{4}$ by too much, we have chosen to limit $\epsilon$ to $\left]0,\frac{1}{10}\right[$. This limitation can be achieved without loss of generality by scaling the problem by a factor of $\frac{1}{5}$. This can be done by introducing a second ancillary qubit and applying a controlled rotation gate $\mathbf{R}$ such that $\mathbf{R}\ket{10}=\frac{2}{\sqrt{5}}\ket{10}+\frac{1}{\sqrt{5}}\ket{11}$. The states of interest are then the ones for which the two ancillary qubits are in state $\ket{11}$. This means that for $\epsilon \in \left]0,\frac{1}{2}\right[$, we are effectively working with $\epsilon' = \frac{\epsilon}{5} \in \left]0,\frac{1}{10}\right[$ as required. With this procedure, the effective error is itself limited to $\left[0,\frac{1}{5}\right]$. This also ensures that we have $\theta_{err} \in \left[0, \arcsin\left(\frac{1}{\sqrt{5}}\right)\right]$ meaning that $\theta_{err} \neq \frac{\pi}{4}$ which is a stationary point of the amplitude amplification process so the error will always be amplified.

To account for this additional procedure, the diffusion operator has to be redefined. We now denote $\mathcal{A}(h) = \big(\mathbf{I}^{\otimes n}\otimes \mathbf{R}\big)\big((\mathbf{TNN}(h)\mathbf{EX}(c,D))\otimes \mathbf{I}\big)$, $\mathcal{X}_G = \mathbf{CZ}(a_0,a_1)$, the controlled $\mathbf{Z}$ gate for which the control is the first ancillary qubit and the target is the second ancillary qubit and $\mathcal{X}_0 = \mathbf{I}-2\ket{0}\bra{0}$. These allow us to define the diffusion operator $\mathbf{Q}$ in a similar way as previously with:
$$
    \mathbf{Q}(h) = - \mathcal{A}(h)\mathcal{X}_0\mathcal{A}^{-1}(h)\mathcal{X}_G
$$
The corresponding circuit is represented in Figure \ref{fig:scale_down} and the tuning algorithm is given in Algorithm \ref{alg:tuning}.
\begin{figure}[h]
    \centering
    $$
    \Qcircuit @C=1em @R=1em{
    \lstick{\ket{0}^{\otimes n}} & \multigate{1}{\mathbf{EX}(c,D)} & \multigate{1}{\mathbf{TNN}(h)} & \qw & \multigate{2}{\mathbf{Q}^m(h)} & \qw \\
    \lstick{\ket{0}^{\phantom{\otimes n}}} & \ghost{\mathbf{EX}(c,D)} & \ghost{\mathbf{TNN}(h)} & \ctrl{1} & \ghost{\mathbf{Q}^m(h)} & \qw \\
    \lstick{\ket{0}^{\phantom{\otimes n}}} & \qw & \qw & \gate{\mathbf{R}} & \ghost{\mathbf{Q}^m(h)} & \qw
    }
    $$
    \caption{Circuit used to tune the network.}
    \label{fig:scale_down}
\end{figure}
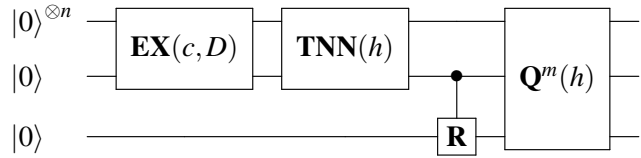

\begin{algorithm}[h]
    \caption{Tuning Algorithm}\label{alg:tuning}
    \KwData{$0< \epsilon < \frac{1}{2}$, $0 < \delta < \frac{1}{2}$, $\mathbf{EX}(c, D)$}
    \KwResult{$\mathbf{TNN}$ expressing $h^*$ such that $P(err_D(h^*,c) < \epsilon) > 1-\delta$}
    $N \gets 2\left(\left\lfloor\frac{1}{\pi\delta^2}\right\rfloor//2\right)+2$\;
    $m_{\text{max}} \gets \text{smallest integer such that } (2m_{\text{max}}+1)\arcsin(\sqrt{\frac{\epsilon}{5}}) \geq \frac{\pi}{4}$\;
    $\mathbf{TNN} \gets \mathbf{I}$\;
    $m \gets -1$\;
    $errors \gets []$\;
    \While{$m < m_{\text{max}}$ \textbf{or} $\text{length}(errors) > \frac{N}{2}$}{
        $m \gets m+1$\;
        $errors \gets []$\;
        $\mathcal{A} \gets \big(\mathbf{I}^{\otimes n}\otimes \mathbf{R}\big)\big((\mathbf{TNN}.\mathbf{EX}(c,D))\otimes \mathbf{I}\big)$\;
        $\ket{\Phi} \gets \mathcal{A}\ket{0}$\;
        $\mathbf{Q} \gets - \mathcal{A}(h)\mathcal{X}_0\mathcal{A}^{-1}(h)\mathcal{X}_G$\;
        \For{$1 \leq n \leq N$}{
            Measure $\mathbf{Q}^m\ket{\Phi}$\;
            \If{11 is measured on the ancillary qubits}{
            append the string of the first $n$ qubits to $errors$
        }
        }
        \If{$length(errors) > \frac{N}{2}$}{
            Update $\mathbf{TNN}$ according to $errors$\;
            $m \gets -1$\;
        }
    }
\end{algorithm}
In Algorithm \ref{alg:tuning} the instruction "Update $\mathbf{TNN}$ according to $errors$" is given without further specification as the way it is done might depend on the class of concepts that is being learnt. The update strategy used to learn the class of concepts introduced in Section \ref{section:particular_class} is given in Algorithm \ref{alg:update}


\section{Proof and Analysis of the Algorithm}
The algorithm stops when the network is in such a state that even after $m_{\text{max}}$ rounds of amplitude amplification, if $S$ is the number of measurements for which the ancillary qubits are in state $\ket{11}$, then $S \leq \frac{N}{2}$ where $N$ is the number of measurements. This stage is only reached if for all the $m$ in the schedule we also have $S \leq \frac{N}{2}$ after $m$ rounds of amplification.
Let $\ket{\Phi} = \mathcal{A}\ket{0}$, then it can be written:
$$
    \ket{\Phi} = \cos(\theta_e)\ket{\phi_{\perp}}+\sin(\theta_e)\ket{\phi_e}\ket{11}
$$
Where $\ket{\phi_{\perp}}$ is orthogonal to $\ket{\phi_e}\ket{11}$ and $\ket{\phi_e}$ contains the inputs that have been misclassified by the network. After $m_{\text{max}}$ rounds of amplification we have:
$$
    \mathbf{Q}^{m_{\text{max}}}\ket{\Phi} = \cos((2m_{\text{max}}+1)\theta_e)\ket{\phi_{\perp}}+\sin((2m_{\text{max}}+1)\theta_e)\ket{\phi_e}\ket{11}
$$
If $p$ is the probability of measuring $11$ on the ancillary qubits then $p = \sin^2((2m_{\text{max}}+1)\theta_e)$ and we want to compare it to $\frac{1}{2}$ in order to apply Lemma \ref{lem:inf}. This boils down to estimating $P\left(p < \frac{1}{2} \ | \ S \leq \frac{N}{2}\right)$.
We have:
$$
    P(p < 1/2 \ | \ S \leq N/2) = \frac{P(S \leq N/2 \ | \ p < 1/2)P(p < 1/2)}{P(S \leq N/2)}
$$
By placing a uniform marginal distribution on $p$ we get:
\begin{equation}\label{eq:Bayes}
    P(p < 1/2 \ | \ S \leq N/2) = \frac{\sum_{k=0}^{N/2}{{N \choose k}\int_0^{\frac{1}{2}}{\theta^k(1-\theta)^{N-k}d\theta}}}{\sum_{k=0}^{N/2}{{N \choose k}\int_0^{1}{\theta^k(1-\theta)^{N-k}d\theta}}}
\end{equation}
From now on we assume that $N$ is even. For $a \in [0,1]$ it can be shown by integrating by parts that:
$$
    \int_0^a{\theta^k(1-\theta)^{N-k}d\theta} = \frac{k!(N-k)!}{(N+1)!} - (1-a)^{N-k+1}\sum_{i=0}^k{\frac{k!(N-k)!}{(k-i)!(N-k+i+1)!}a^{k-i}(1-a)^i}
$$
This leads to:
$$
\int_0^1{\theta^k(1-\theta)^{N-k}d\theta} = \frac{k!(N-k)!}{(N+1)!}
$$
And
\begin{equation}\label{eq:0to1}
    \sum_{k=0}^{N/2}{{N \choose k}\int_0^{1}{\theta^k(1-\theta)^{N-k}d\theta}} = \frac{N+2}{2(N+1)}
\end{equation}
Similarly we have:
$$
    \int_0^{\frac{1}{2}}{\theta^k(1-\theta)^{N-k}d\theta} = \frac{k!(N-k)!}{(N+1)!} - \left(\frac{1}{2}\right)^{N+1}\sum_{i=0}^k{\frac{k!(N-k)!}{(k-i)!(N-k+i+1)!}}
$$
And
\begin{equation}\label{eq:0to1/2}
    \sum_{k=0}^{N/2}{{N \choose k}\int_0^{\frac{1}{2}}{\theta^k(1-\theta)^{N-k}d\theta}} = \frac{N+2}{2(N+1)} - \left(\frac{1}{2}\right)^{N+1}\sum_{k=0}^{N/2}{\sum_{i=0}^k{\frac{N!}{(k-i)!(N-k+i+1)!}}}
\end{equation}
Putting (\ref{eq:Bayes}), (\ref{eq:0to1}) and (\ref{eq:0to1/2}) together yields:
\begin{align}
     P(p < 1/2 \ | \ S \leq N/2) &= 1 - \left(\frac{1}{2}\right)^{N}\frac{N+1}{N+2}\sum_{k=0}^{N/2}{\sum_{i=0}^k{\frac{N!}{(k-i)!(N-k+i+1)!}}} \notag\\
     &= 1 - \left(\frac{1}{2}\right)^{N}\frac{1}{N+2}\sum_{k=0}^{N/2}{\sum_{i=0}^k{\frac{(N+1)!}{(k-i)!(N-k+i+1)!}}} \notag\\
     &= 1 - \left(\frac{1}{2}\right)^{N}\frac{1}{N+2}\sum_{k=0}^{N/2}{\sum_{i=0}^k{{N+1 \choose k-i}}} \notag\\
     &= 1 - \left(\frac{1}{2}\right)^{N}\frac{1}{N+2}\sum_{k=0}^{N/2}{\sum_{i=0}^k{{N+1 \choose i}}} \notag\\
     &= 1 - \left(\frac{1}{2}\right)^{N}\frac{1}{N+2}\sum_{k=0}^{N/2}{{N+1 \choose k}\left(\frac{N}{2}+1-k\right)}\label{eq:proba1}
\end{align}
Now
\begin{equation}\label{eq:inter1}
    \sum_{k=0}^{N/2}{{N+1 \choose k}\left(\frac{N}{2}+1-k\right)} = \frac{N+2}{2}\sum_{k=0}^{N/2}{{N+1 \choose k}} - \sum_{k=0}^{N/2}{{N+1 \choose k}k}
\end{equation}
For $k>0$ we have:
$$
    {N+1 \choose k}k = {N \choose k-1}(N+1)
$$
And $N$ being even, we also have:
$$
    \sum_{k=0}^{N/2}{{N+1 \choose k}} = 2^N
$$
Plugging these back into (\ref{eq:inter1}) we get:
$$
    \sum_{k=0}^{N/2}{{N+1 \choose k}\left(\frac{N}{2}+1-k\right)} = 2^{N-1}(N+2) - (N+1)\sum_{k=0}^{\frac{N}{2}-1}{{N \choose k}}
$$
Coming back to (\ref{eq:proba1}), we thus have:
\begin{equation}\label{eq:proba2}
    P(p < 1/2 \ | \ S \leq N/2) = \frac{1}{2} + \left(\frac{1}{2}\right)^N\frac{N+1}{N+2}\sum_{k=0}^{\frac{N}{2}-1}{{N \choose k}}
\end{equation}
Because $N$ is even:
$$
    \sum_{k=0}^{\frac{N}{2}-1}{{N \choose k}} = \frac{1}{2}\left(2^N - {N \choose N/2}\right)
$$
Together with (\ref{eq:proba2}), this leads to:
$$
    P(p < 1/2 \ | \ S \leq N/2) = \frac{1}{2} + \left(\frac{1}{2}\right)^{N+1}\frac{N+1}{N+2}\left(2^N - {N \choose N/2}\right)
$$
But \cite{Cover2005}:
$$
    {N \choose N/2} \leq \sqrt{2}\frac{2^N}{\sqrt{\pi N}}
$$
So:
$$
    P(p < 1/2 \ | \ S \leq N/2) \geq \frac{1}{2} + \frac{N+1}{2(N+2)}\left(1-\frac{\sqrt{2}}{\sqrt{\pi N}}\right)
$$
Because $\frac{N+1}{N+2} \sim 1$ we look for a lower bound of the form $1-\frac{\alpha}{\sqrt{N}}$. For example, we have:
$$
    P(p < 1/2 \ | \ S \leq N/2) > 1 - \frac{1}{\sqrt{\pi N}}
$$
So for $0 < \delta < \frac{1}{2}$, in order to have $P(p < 1/2 \ | \ S \leq N/2) > 1 - \delta$, it suffices to take for $N$ an even integer that is greater than $\frac{1}{\pi \delta^2}$, hence:
$$
    N = 2\left(\left\lfloor\frac{1}{\pi \delta^2}\right\rfloor//2\right)+2
$$
We have shown that when the algorithm stops, we have:
$$
    P\left(\sin^2\left((2m_{\text{max}}+1)\theta_e\right) < \frac{1}{2}\right) > 1 - \delta
$$
Together with Lemma \ref{lem:inf}, it comes that:
$$
    P(\theta_e < \theta_{\epsilon}) > 1 - \delta 
$$
Hence:
$$
    P(err < \epsilon) > 1- \delta
$$
So each sampling phase requires $\frac{1}{\pi \delta^2}$ calls to the oracle $\mathbf{EX}(c,d)$. In order to perform one update of the network, the algorithm requires at most $m_{\text{max}} \approx \frac{1}{2}\left(\frac{\pi}{4\arcsin(\sqrt{\epsilon/5})}-1\right)$ of these sampling phases. Now for $\epsilon \in \left]0, \frac{1}{2}\right[$, we have
$\arcsin(\sqrt{\epsilon/5}) \approx \sqrt{\epsilon/5}$ so for one update of the network, the total number of call to $\mathbf{EX}(c,d)$ is $O\left(\frac{1}{\delta^2}\frac{1}{\sqrt{\epsilon}}\right)$. The number of updates needed to reach the learning target is dependant on the class of concept. 

It is possible to reduce the number of calls to $ \mathbf{EX}(c,d)$ during one update phase by incrementing $m$ not by 1 but with powers of a given number. In this case the total number for an update is $O\left(\frac{1}{\delta^2}\log\left(\frac{1}{\epsilon}\right)\right)$ but this comes at the cost of possibly having a lower probability of success as we will see in Section \ref{section:particular_class}. 


\section{Learning a Particular Class}\label{section:particular_class}
Let $n \in \N$, for $s \in \B^n$, we define the parity function $p_s : x \mapsto s.x = s_0.x_0 \oplus \ldots \oplus s_{n-1}.x_{n-1}$ and we are interested in learning the class of the parity functions $\mathcal{C}_{p}$:
$$
    \mathcal{C}_{p} = \left\{p_s \ | \ s \in \B^n\right\}
$$
Any concept from this class can easily be expressed by a TNN: for each non-zero bit of $s$ it suffices to apply the $\mathbf{X}$ gate controlled by the corresponding qubit. In order to learn a parity function, we are applying the update strategy shown in Algorithm \ref{alg:update}. Because the gates to be updated are all controlled by a single qubit, we are assured that the final hypothesis will be a parity function.
\begin{algorithm}[h]
    \caption{Update strategy}\label{alg:update}
    \KwData{$errors$ and $corrects$ the lists of measured inputs that are respectively misclassified and correctly classified}
    \KwResult{$gates$ the list of gates to be updated}
    Gather the measurements in $errors$ by group of same Hamming weight so that $errors[i]$ is the list of misclassified inputs of Hamming weight $i$\;
    Do the same with $corrects$\;
    $gates \leftarrow errors[1]$\;
    \For{$1 \leq i \leq n-1$}{
        \For{$e$ in $errors[i]$}{
            \For{$c$ in $corrects[i+1]$}{
                \If{$e \oplus c$ has Hamming weight 1}{
                    Add $e \oplus c$ to $gates$ if it is not already in\;
                }
            }
        }
        \For{$c$ in $corrects[i]$}{
            \For{$e$ in $errors[i+1]$}{
                \If{$c \oplus e$ has Hamming weight 1}{
                    Add $c \oplus e$ to $gates$ if it is not already in\;
                }
            }
        }
    }
    Return $gates$\;
\end{algorithm}

This approach has been implemented for $n=4,6$ and 8 using the Qiskit library\footnote{The code can be found in the following repository: https://github.com/vietphamngoc/QPAC}. For a given $n$, a probability distribution over $\B^n$ has been created randomly by applying to each qubit a $\mathbf{R}_y$ rotation gate with an angle randomly chosen in $[0, \pi]$. The controlled $\mathbf{X}$ gates corresponding to the target concept are then added to the circuit. These two blocks taken together are constituting the oracle $\mathbf{EX}(c,D)$. The construction of such oracles for $n=4$ is illustrated in Figure \ref{fig:oracle}. The $\mathbf{TNN}$ has then been trained to learn the concepts of the class $\mathcal{C}_p$ according to the training algorithm and the update strategy introduced in this paper. For $n=4$ the experiments have been run for all the concepts of the class while for $n=6$ and 8 the experiments have been performed for 16 randomly selected concepts from this class. In all cases different values of $\epsilon$ and $\delta$. Once the training has been completed, the error is evaluated using the Qiskit statevector simulator by running the circuit composed of the oracle and the tuned network and getting the amplitude of all the inputs for which the network's output is wrong. 
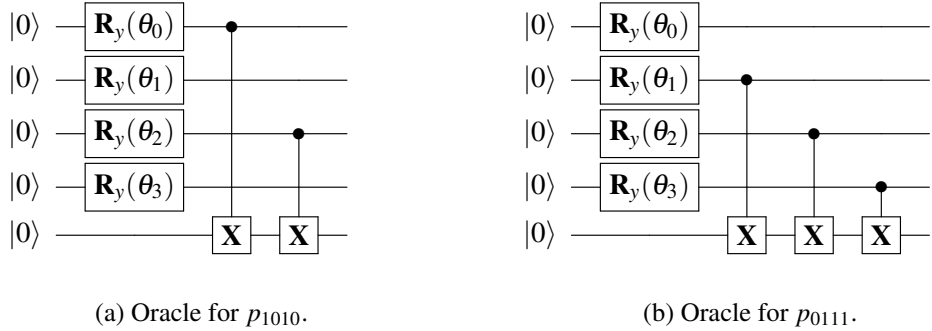
\begin{figure}[h]
    \centering
    \begin{subfigure}{0.45\textwidth}
        \centering
        $$
        \Qcircuit @C=1em @R=0.2em{
        \lstick{\ket{0}} & \gate{\mathbf{R}_y(\theta_0)}    & \ctrl{4}          & \qw               & \qw\\
        \lstick{\ket{0}} & \gate{\mathbf{R}_y(\theta_1)}    & \qw               & \qw               & \qw\\
        \lstick{\ket{0}} & \gate{\mathbf{R}_y(\theta_2)}    & \qw               & \ctrl{2}          & \qw\\
        \lstick{\ket{0}} & \gate{\mathbf{R}_y(\theta_3)}    & \qw               & \qw               & \qw\\
        \lstick{\ket{0}} & \qw                              & \gate{\mathbf{X}} & \gate{\mathbf{X}} & \qw
        }
        $$
        \caption{Oracle for $p_{1010}$.}
    \end{subfigure}
    \begin{subfigure}{0.45\textwidth}
        \centering
         $$
        \Qcircuit @C=1em @R=0.2em{
        \lstick{\ket{0}} & \gate{\mathbf{R}_y(\theta_0)}    & \qw               & \qw               & \qw               & \qw\\
        \lstick{\ket{0}} & \gate{\mathbf{R}_y(\theta_1)}    & \ctrl{3}          & \qw               & \qw               & \qw\\
        \lstick{\ket{0}} & \gate{\mathbf{R}_y(\theta_2)}    & \qw               & \ctrl{2}          & \qw               & \qw\\
        \lstick{\ket{0}} & \gate{\mathbf{R}_y(\theta_3)}    & \qw               & \qw               & \ctrl{1}          & \qw\\
        \lstick{\ket{0}} & \qw                              & \gate{\mathbf{X}} & \gate{\mathbf{X}} & \gate{\mathbf{X}} & \qw
        }
        $$
        \caption{Oracle for $p_{0111}$.}
    \end{subfigure}
    \caption{Implementation of the oracle for $n=4$ and different target concepts.}
    \label{fig:oracle}
\end{figure}

To account for the overall randomness of the process, each training has been performed 50 times and the results are represented with violin plots where the width of the plot is proportional to their distribution within the repetitions. In Figure \ref{fig:errors} are represented the results of some of these experiments. For given function and $\epsilon$, in general, as $\delta$ decreases, the distribution of the final error rates tend to accumulates toward a minimum value, which was expected. Regarding the maximum value, we notice that in our set up, it never went above $\epsilon$, except for the case $n=6$, $\epsilon=0.05$ and $\delta=0.1$ for the function $p_{010001}$ as can be seen in Figure \ref{fig:errors_6_0.05}. In this specific case the final error rate was greater than $\epsilon$ once in the 50 experiments. This event did not happen again for lower values of $\delta$. 

\begin{figure}[h]
    \centering
    \begin{subfigure}{0.49\textwidth}
        \centering
        \includegraphics[width=\textwidth]{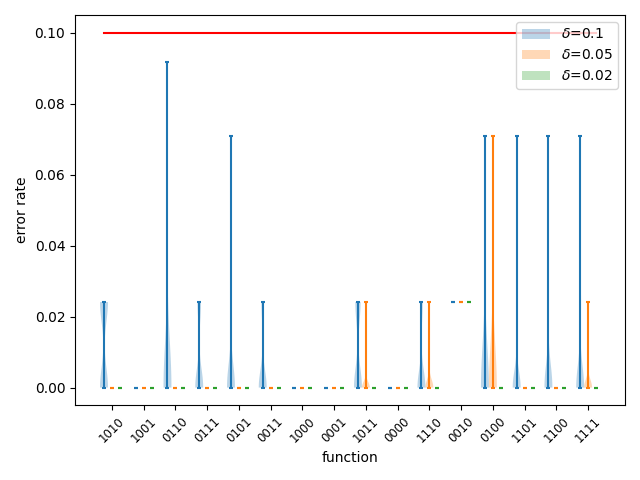}
        \caption{$n=4$ and $\epsilon=0.1$}
        \label{fig:errors_4_0.1}
    \end{subfigure}
    \hfill
    \begin{subfigure}{0.49\textwidth}
        \centering
        \includegraphics[width=\textwidth]{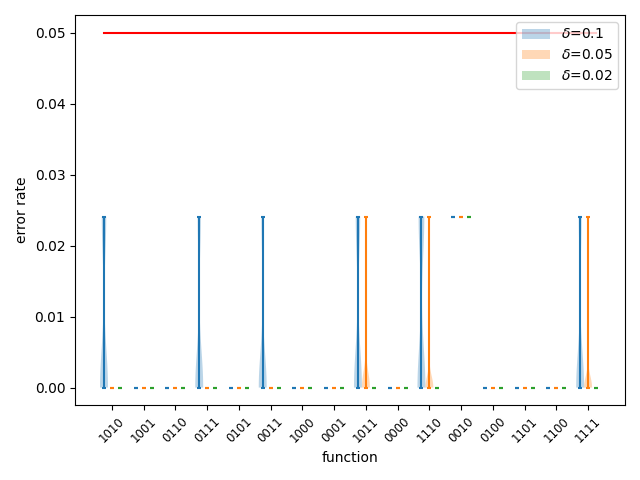}
        \caption{$n=4$ and $\epsilon=0.05$}
        \label{fig:errors_4_0.05}
    \end{subfigure}

    \begin{subfigure}{0.49\textwidth}
        \centering
        \includegraphics[width=\textwidth]{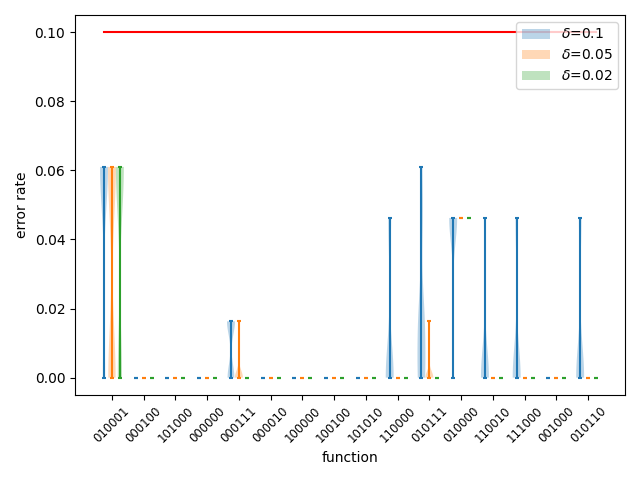}
        \caption{$n=6$ and $\epsilon=0.1$}
        \label{fig:errors_6_0.1}
    \end{subfigure}
    \hfill
    \begin{subfigure}{0.49\textwidth}
        \centering
        \includegraphics[width=\textwidth]{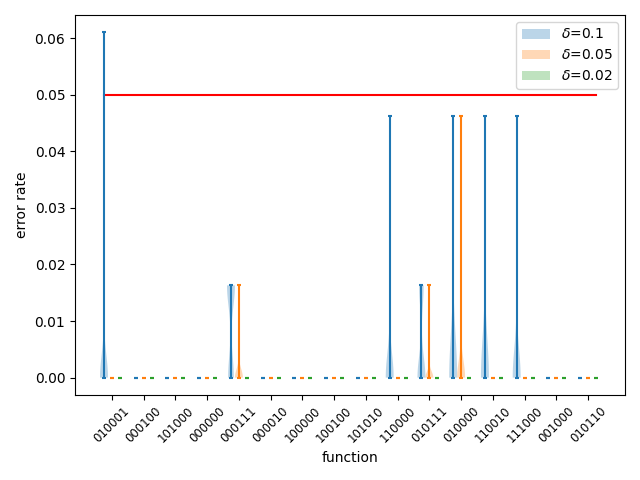}
        \caption{$n=6$ and $\epsilon=0.05$}
        \label{fig:errors_6_0.05}
    \end{subfigure}

    \begin{subfigure}{0.49\textwidth}
        \centering
        \includegraphics[width=\textwidth]{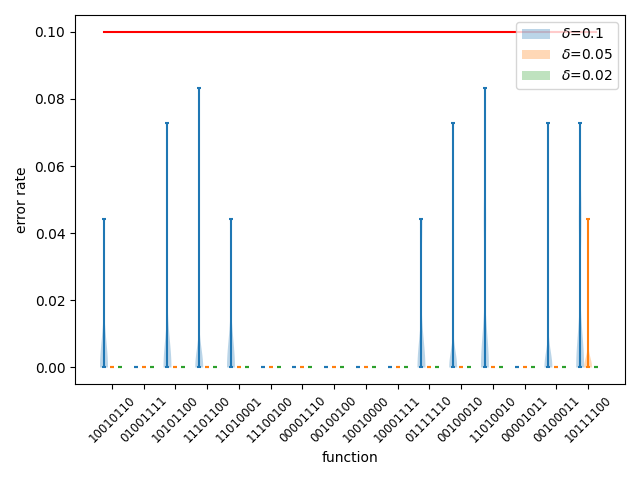}
        \caption{$n=8$ and $\epsilon=0.1$}
        \label{fig:errors_8_0.1}
    \end{subfigure}
    \hfill
    \begin{subfigure}{0.49\textwidth}
        \centering
        \includegraphics[width=\textwidth]{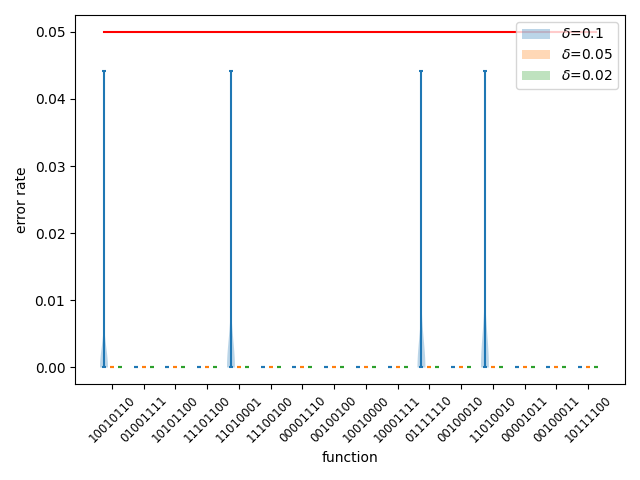}
        \caption{$n=8$ and $\epsilon=0.05$}
        \label{fig:errors_8_0.05}
    \end{subfigure}
    
    \caption{Final error rates for different values of $n$, $\epsilon$ and $\delta$. The red line represents $\epsilon$.}
    \label{fig:errors}
\end{figure}

As said previously, the number of update steps necessary until a state with an error rate lower than $\epsilon$ depends on the class of concept. In the case of $\mathcal{C}_p$ and with the experiments realised, we have found that the number of update steps decreases with decreasing $\delta$. This can be explained by the fact that decreasing $\delta$ will increase the number of samples, hence the diversity of the measurements. This in turn allows for more accurate updates which result in less updates. This trend is depicted in Figure \ref{fig:updates} where the number of updates is plotted in the cases $n = 4$ and 8, $\epsilon=0.05$ and for different values of $\delta$. Although there are few experiments, we can see that the number of updates does not significantly increases with $n$.   
\begin{figure}[ht]
    \centering
    \begin{subfigure}{0.49\textwidth}
        \centering
        \includegraphics[width=\textwidth, height=5.5cm]{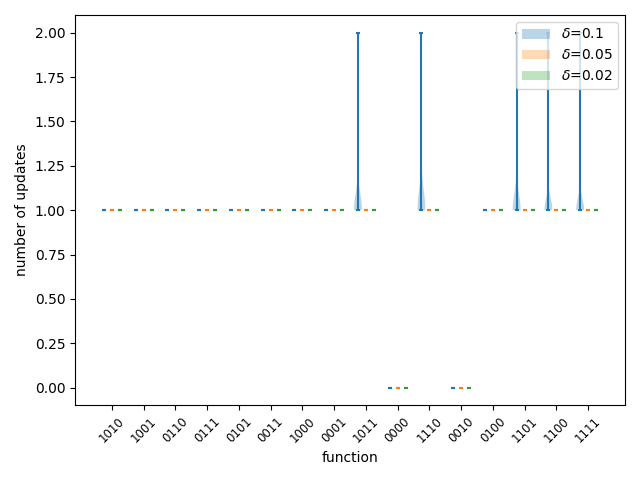}
        \caption{$n=4$}
        \label{fig:updates_4}
    \end{subfigure}
    \hfill
    \begin{subfigure}{0.49\textwidth}
        \centering
        \includegraphics[width=\textwidth, height=5.5cm]{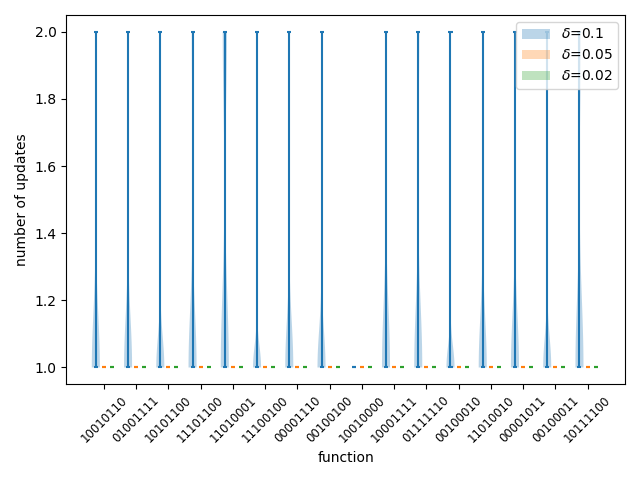}
        \caption{$n=8$}
        \label{fig:updates_8}
    \end{subfigure}
    \caption{Number of updates for different values of $n$ and $\epsilon=0.05$}
    \label{fig:updates}
\end{figure}

Finally, we wished to investigate how the increment schedule of $m$, the number of amplification round, could impact the behaviour of the tuning algorithm. In Algorithm \ref{alg:tuning}, $m$ was incremented by 1 after each sampling round. To speed-up the algorithm, the increment schedule was changed to increase $m$ with powers of 2. We have compared the two schedule for $\epsilon=0.01$, that is $m \in [0,9]$ against $m \in \{0,1,2,4,8,9\}$, in the case $n=6$ and with different values of $\delta$. As previously, for each different value of $\delta$ the experiments have been repeated 50 times. The results of these experiments, as reported in Figure \ref{fig:steps}, show that incrementing $m$ with powers of 2 yields similar performances compared to when $m$ is incremented by 1. This means that the total number of calls to the oracle could be reduced from $O\left(\frac{1}{\delta^2}\frac{1}{\sqrt{\epsilon}}\right)$ to $O\left(\frac{1}{\delta^2}\ln\left(\frac{1}{\epsilon}\right)\right)$. 

\begin{figure}[ht]
    \centering
    \begin{subfigure}{0.49\textwidth}
        \centering
        \includegraphics[width=\textwidth, height=5.5cm]{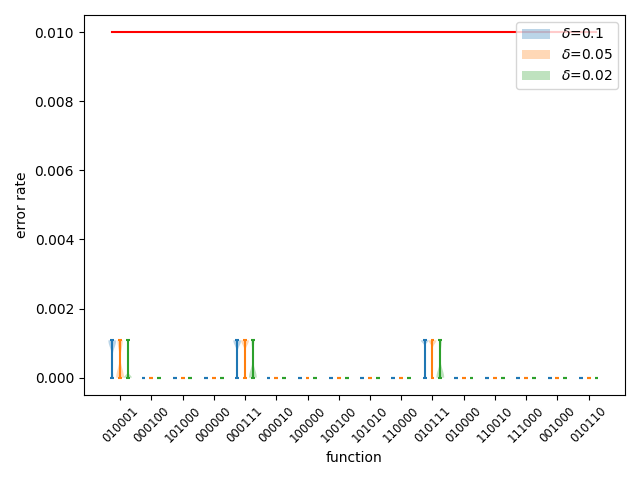}
        \caption{$m$ incremented by 1}
        \label{fig:errors_1}
    \end{subfigure}
    \hfill
    \begin{subfigure}{0.49\textwidth}
        \centering
        \includegraphics[width=\textwidth, height=5.5cm]{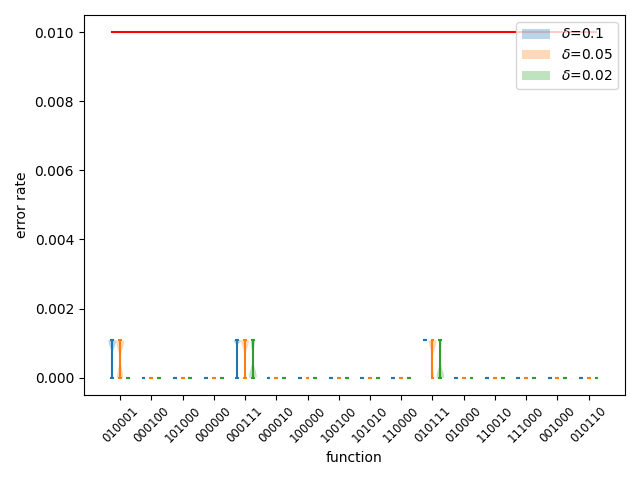}
        \caption{$m$ incremented with powers of 2}
        \label{fig:errors_2}
    \end{subfigure}
    \begin{subfigure}{0.49\textwidth}
        \centering
        \includegraphics[width=\textwidth, height=5.5cm]{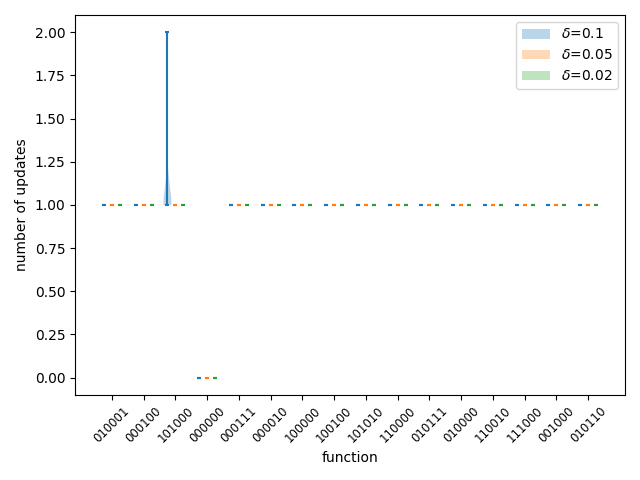}
        \caption{$m$ incremented by 1}
        \label{fig:updates_1}
    \end{subfigure}
    \hfill
    \begin{subfigure}{0.49\textwidth}
        \centering
        \includegraphics[width=\textwidth, height=5.5cm]{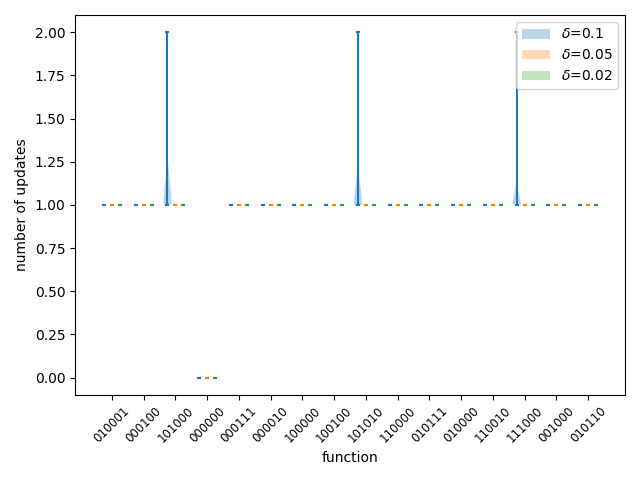}
        \caption{$m$ incremented with powers of 2}
        \label{fig:updates_2}
    \end{subfigure}
    \caption{$n=6$ and $\epsilon=0.01$. Final error rate (Figures \ref{fig:errors_1} and \ref{fig:errors_2}) and number of updates (Figures \ref{fig:updates_1} and \ref{fig:updates_2}), for different values of $\delta$, resulting from different increment schedules for $m$.}
    \label{fig:steps}
\end{figure}

\section{Conclusion}
In this paper we have studied tunable quantum neural networks in the context of QPAC-learning. To do so, we have devised and proved a learning algorithm that uses quantum amplitude amplification. Amplitude amplification is used to both compare the error rate to the threshold $\epsilon$ and to better measure the errors. These measurements are then used to update the network. We have implemented this approach and tested it against the class of parity functions and found that this algorithm is indeed an efficient learner as its sample complexity is $O\left(\frac{1}{\delta^2}\frac{1}{\sqrt{\epsilon}}\right)$ with a possible reduction to $O\left(\frac{1}{\delta^2}\ln\left(\frac{1}{\epsilon}\right)\right)$. Experimental results show that most of the final error rates are quite far away from the threshold $\epsilon$, which could be explained by an excessive number of samples. As future work we will look for a tighter lower bound for the probability of success which should reduce the dependence in $\frac{1}{\delta}$ in the algorithm's complexity. Finally we will study the generalisation of this approach to more complex classes of concepts.

\nocite{*}
\bibliographystyle{eptcs}
\bibliography{generic}
\end{document}